\documentclass[11pt]{article}

\newcommand{\aop}{Ann. Phys.~}
\newcommand{\cmp}{Comm. Math. Phys.~}

\newcommand{\jmp}{J. Math. Phys.~}
\newcommand{\jpa}{J. Phys. A~}

\newcommand{\prl}{Phys. Rev. Lett.~}

\newcommand{\laa}{Lin. Alg. App.~}

\usepackage{appendix}
\usepackage{epstopdf}
\usepackage{color}
\usepackage{subfigure}
\usepackage[T1]{fontenc}
\usepackage[sc]{mathpazo}
\usepackage{amsmath}
\usepackage{amssymb}
\usepackage{graphicx,color}
\usepackage{framed}
\usepackage{multirow}
\usepackage{enumerate}
\usepackage{amsthm}
\usepackage{amsfonts,mathrsfs}
\usepackage{geometry} 
\usepackage{eepic}
\usepackage{ifthen}
\usepackage[vcentermath]{youngtab}
\usepackage[unicode=true,pdfusetitle, bookmarks=true,bookmarksnumbered=false,bookmarksopen=false, breaklinks=false,pdfborder={0 0 0},backref=false,colorlinks=false] {hyperref}
\hypersetup{
colorlinks,linkcolor=myurlcolor,citecolor=myurlcolor,urlcolor=myurlcolor}
\definecolor{myurlcolor}{rgb}{0,0,0.7}

\usepackage{color}

\newcommand{\blue}{\textcolor{blue}}

\usepackage{hyperref}
\hypersetup{pdfpagemode=UseNone}

\newcommand{\tinyspace}{\mspace{1mu}}

\newcommand{\op}[1]{\operatorname{#1}}

\newcommand{\abs}[1]{\left\lvert\tinyspace #1 \tinyspace\right\rvert}

\renewcommand{\det}{\operatorname{det}}

\newcommand{\setft}[1]{\mathrm{#1}}

\newcommand{\density}[1]{\setft{D}\left(#1\right)}

\newcommand{\sign}{\op{sign}}

\def\liet{\mathfrak{t}}

\def\vol{\mathrm{vol}}
\def\dh{\mathrm{DH}}

\def \dif {\mathrm{d}}
\def \diag {\mathrm{diag}}
\def \vol {\mathrm{vol}}
\def \re {\mathrm{Re}}

\def\complex{\mathbb{C}}
\def\real{\mathbb{R}}
\def\natural{\mathbb{N}}
\def\integer{\mathbb{Z}}
\def\I{\mathbb{1}}

\newenvironment{mylist}[1]{\begin{list}{}{
    \setlength{\leftmargin}{#1}
    \setlength{\rightmargin}{0mm}
    \setlength{\labelsep}{2mm}
    \setlength{\labelwidth}{8mm}
    \setlength{\itemsep}{0mm}}}
    {\end{list}}

\newcommand{\inner}[2]{\langle #1 , #2\rangle}

\newcommand{\Inner}[2]{\left\langle #1 , #2\right\rangle}

\newcommand{\defeq}{\stackrel{\smash{\textnormal{\tiny def}}}{=}}

\newcommand{\Pa}[1]{\left(#1\right)}

\newcommand{\Br}[1]{\left[#1\right]}

\newcommand{\Set}[1]{\left\{#1\right\}}

\DeclareMathOperator{\trace}{Tr}

\newcommand{\Ptr}[2]{\trace_{#1}\Pa{#2}}

\newcommand{\Tr}[1]{\Ptr{}{#1}}

\def\cI{\mathcal{I}}
\def\cO{\mathcal{O}}

\def\cU{\mathcal{U}}

\def\bsA{\boldsymbol{A}}\def\bsB{\boldsymbol{B}}\def\bsC{\boldsymbol{C}}
\def\bsH{\boldsymbol{H}}

\def\bsS{\boldsymbol{S}}
\def\bsU{\boldsymbol{U}}\def\bsV{\boldsymbol{V}}\def\bsW{\boldsymbol{W}}\def\bsX{\boldsymbol{X}}\def\bsY{\boldsymbol{Y}}
\def\bsZ{\boldsymbol{Z}}

\def\bsa{\boldsymbol{a}}\def\bsb{\boldsymbol{b}}\def\bsc{\boldsymbol{c}}

\def\bss{\boldsymbol{s}}
\def\bsu{\boldsymbol{u}}\def\bsx{\boldsymbol{x}}

\def\rD{\mathrm{D}}\def\rE{\mathrm{E}}
\def\rG{\mathrm{G}}\def\rH{\mathrm{H}}\def\rJ{\mathrm{J}}

\def\rQ{\mathrm{Q}}\def\rS{\mathrm{S}}
\def\rU{\mathrm{U}}

\def\sC{\mathscr{C}}

\def\T{\textsf{T}}

\newtheorem{thrm}{Theorem}[section]
\newtheorem{lem}[thrm]{Lemma}
\newtheorem{prop}[thrm]{Proposition}
\newtheorem{cor}[thrm]{Corollary}
\theoremstyle{definition}

\newtheorem{remark}[thrm]{Remark}
\newtheorem{exam}[thrm]{Example}

\numberwithin{equation}{section}

\newcounter{questionnumber}

\begin{document}

\title{A variant of Horn's problem and the derivative principle}

\author{\blue{Lin Zhang}$^{1,2}$\footnote{E-mail: godyalin@163.com; linzhang@mis.mpg.de},\quad \blue{Hua Xiang}$^3$\footnote{E-mail: hxiang@whu.edu.cn}\\
  {\it\small $^1$Institute of Mathematics, Hangzhou Dianzi University, Hangzhou 310018, PR~China}\\
  {\it\small $^2$Max-Planck-Institute for Mathematics in the Sciences, Leipzig 04103, Germany}\\
  {\it\small $^3$School of Mathematics and Statistics, Wuhan University, Wuhan 430072, PR~China}}

\date{}
\maketitle \mbox{}\hrule\mbox{}
\begin{abstract}

Identifying the spectrum of the sum of two given Hermitian matrices
with fixed eigenvalues is the famous Horn's problem.
In this note, we investigate a variant of Horn's problem, i.e., we
identify the probability density function (abbr. pdf) of the
diagonals of the sum of two random Hermitian matrices with given
spectra. We then use it to re-derive the pdf of the eigenvalues of
the sum of two random Hermitian matrices with given eigenvalues via the
\emph{derivative principle}, a powerful tool used to get the exact
probability distribution by reducing to the corresponding
distribution of diagonal entries.
We can also recover Jean-Bernard Zuber's recent results.
Moreover, as an illustration, we derive the analytical expressions
of eigenvalues of the sum of two random Hermitian matrices from
$\rG\rU\rE(n)$ or Wishart ensemble by the derivative principle,
respectively.
We also investigate the statistics of exponential of random matrices
and connect them with Golden-Thompson inequality, and partially answer
a question proposed by Forrester.
Some potential applications in quantum information theory, such as
uniform average quantum Jensen-Shannon divergence and average
coherence of uniform mixture of two orbits,
are discussed.

\end{abstract}
\mbox{}\hrule\mbox{}

\section{Introduction}

The famous Horn's problem asks for the spectrum of the sum of two
given Hermitian matrices with fixed eigenvalues. Specifically,
Horn's problem characterizes the triple $(\bsa,\bsb,\bsc)$ for which
there exist Hermitian matrices $\bsA,\bsB,\bsC$ with respective
eigenvalues $\bsa=(a_1,\ldots,a_n),\bsb=(b_1,\ldots,b_n)$ and
$\bsc=(c_1,\ldots,c_n)$ satisfying the following constraint:
$$
\bsA+\bsB=\bsC.
$$
This problem has an affirmative answer \cite{Cao2012}, in terms of
linear inequalities which are now called \emph{Horn's inequalities}.
The solutions form a convex polytope whose describing inequalities
have been conjectured by Horn in 1962 \cite{Horn1962}. Note that the
convex polytope for the solution of Horn's problem is, in general,
nontrivial. Hence each point in this convex polytope corresponds to
a possible eigenvalue vector for the sum, except the trivial cases
(for example, when one of the matrices is scalar).

Although Horn's problem is apparently an elementary problem (a
complete answer to Horn's problem takes almost a century), it turns
out to be connected with many areas of mathematics: linear algebra
of course \cite{Knutson}, but also combinatorics, algebraic geometry
\cite{Knutson}, symplectic geometry, and even probability theory,
etc. For instance, Alekseev \emph{et. al} give a symplectic proof of
the Horn inequalities on eigenvalues of a sum of two Hermitian
matrices with given spectra \cite{Alekseev2014}, and Zuber
investigate the probability distribution over the Horn's polytope
\cite{JBZ2017}. Related probabilistic works include
\cite{Faraut2019,Forrester2019}.

Besides, some researchers give a description of the
Duistermaat-Heckman measure on the Horn polytope.
Mathematical speaking, the eigenvalue distributions involved are
so-called \emph{Duistermaat-Heckman measures} \cite{Christandl2014},
which are defined using the push-forward of the Liouville measure on
a symplectic manifold along the moment map. We follow the notations
along the paper \cite{Christandl2014}, and consider the problem of
describing the sum of two coadjoint orbits $\cO_{\bsa}+\cO_{\bsb}$,
where $K$ acts on $\cO_{\bsa}\times \cO_{\bsb}$ diagonally with
moment map $(\bsA,\bsB)\mapsto \bsA+\bsB$ and we have
\cite{Christandl2014}: Let $\bsa\in\liet^*_{>0}$ and $\bsb\in
\liet^*_{\geqslant0}$. Then $\dh^K_{\cO_{\bsa}\times\cO_{\bsb}} =
\sum_{w\in W}(-1)^{l(w)}\delta_{w\bsa}*\dh^T_{\cO_{\bsb}}$, where
$l(w)$ is the length of the Weyl group element $w$; and $\dh$ is the
Duistermaat-Heckamn measure.
Apparently, this result gives the distribution of solutions of
Horn's problem in theoretical level completely, but however in
specific computational problems, it is less useful. In this paper,
instead of Horn's polytope itself, we will consider the pdf of the
diagonals of sum of two Hermitian matrices with prescribed spectra.
Then by employing derivative principle, we obtain the probability
distribution density of the eigenvalues of sum of two Hermitian
matrices with prescribed spectra. The support for such probability
distribution density function is just Horn's polytope, determined by
Horn's inequalities. The obtained pdfs for the diagonals and
spectra, respectively, are expressed by complicated complex
integrals that can be explicitly calculated in lower dimensional
cases, in particular for $2\times 2$ case. We also apply this
special case to analyze some quantities used in quantum information
theory.
We remark that Horn's problem is subsumed into the one-body quantum
marginal problem, i.e., the problem of determining the set of
possible reduced density matrices, known as the quantum marginal
problem in quantum information theory and as the
$N$-representability problem in quantum chemistry.

The paper is organized as follows. In
Sect.~\ref{sect:pdf-diagonals}, we derive an analytical formula for
the pdf of the diagonals of sum of two random Hermitian matrices
with prescribed spectra. The main results are summarized into
Theorem~\ref{thrm:main} and Corollary~\ref{cor:main}.
In Sect.~\ref{sect:der-pin}, we firstly recall the derivative
principle, and then combine this with the pdf of diagonals to derive
the pdf of eigenvalues of sum of two random Hermitian matrices (see
Theorem~\ref{thrm:CorollaryDP}).
We also find 
that the results obtained in \cite{JBZ2017} can be naturally derived
from our main result. These materials are summarized into the
Appendix~\ref{app:A}.
We further derive the analytical expressions for eigenvalues of the
sum of two random Hermitian matrices from $\rG\rU\rE(n)$ or Wishart
ensemble by derivative principle, respectively, which are summarized
in Propositions~\ref{th:sum-of-Kgue}.
Sequentially, in Sect.~\ref{sect:GTinequality}, we use the pdf of
the sum of two GUE random matrices to derive the expectation of the
matrix exponential and partly answer a question proposed by
Forrester \cite{Forrester2014}.
As another application of our results, in Sect.~\ref{sect:app}, we
use our lower dimensional result to analyze some quantities used in
quantum information theory. Finally, we conclude with some remarks
in Sect.~\ref{sect:end}.

\section{The pdf of diagonals of the sum of two random Hermitian matrices with given spectra} \label{sect:pdf-diagonals}

We present an explicit formula concerning the joint distribution
density of the diagonals of the sum of two random Hermitian matrices
with given spectra. This is a new result in this note, and it will
be used to derive the pdf of eigenvalues of the sum of two random
Hermitian matrices with given spectra.
To keep accordance with the notation in the present literature, the notation adopted here is only a little different from that in
\cite{Mejia2017}.

Denote $\bsx=(x_1,\ldots,x_n)$, $\widehat{\bsx} =
\diag(x_1,\ldots,x_n)$ and $\Delta(\bsx)=\prod_{1\leqslant
i<j\leqslant n}(x_i-x_j),[\dif\bsx]=\prod^n_{j=1}\dif x_j$. We also
denote the unitary orbit with spectrum $\bsx$ by
$\cU(\bsx)=\Set{\bsU\widehat \bsx \bsU^\dagger: \bsU\in\rU(n)}$,
where $\rU(n)$ stands for the $n\times n$ unitary matrix group.
\begin{thrm}\label{thrm:main}
Assume that two random matrices $\bsA$ and $\bsB$ chosen uniformly
on the unitary orbits $\cU(\bsa)$ and $\cU(\bsb)$, respectively, then the
joint pdf $q(\bsC^\diag|\bsa,\bsb)$ of the diagonal part
$\bsC^\diag$ of the sum $\bsC=\bsA+\bsB$ is given by the following
integral
\begin{eqnarray}
q(\bsC^\diag|\bsa,\bsb) &=&\frac M{\Delta(\bsa)\Delta(\bsb)}
\int_{\real^n}[\dif\bsx]\frac{\det\Pa{e^{\mathrm{i}x_ia_j}}\det\Pa{e^{\mathrm{i}x_ib_j}}}{\Delta(\bsx)^2\prod^n_{i=1}e^{\mathrm{i}x_i\bsC_{ii}}},
\end{eqnarray}
where
\begin{eqnarray}\label{eq:constM}
M =
\frac{\Pa{\prod^n_{k=1}\Gamma(k)}^2}{(2\pi)^n\mathrm{i}^{n(n-1)}}.
\end{eqnarray}
\end{thrm}

\begin{proof}
From \cite{JBZ2017}, we see that the pdf of $\bsC=\bsU\widehat{\bsa}
\bsU^\dagger + \bsV\widehat{\bsb} \bsV^\dagger$ is given by
\footnote{Note that the factor $2^{-n}\pi^{-n^2}$ is different
from that used by Zuber, i.e., $(2\pi)^{-n^2}$. 
There exists a symmetry condition about $X$, that is,
Hermiticity of $X$, and $\delta(X)=2^{-n}\pi^{-n^2}\int
\exp(\mathrm{i}\Tr{TX})[\dif T]$.}
\begin{eqnarray*}
p(\bsC|\bsa,\bsb) = \frac1{2^n\pi^{n^2}} \int [\dif \bsX]
e^{-\mathrm{i}\Tr{\bsX\bsC}}\int_{\rU(n)}\int_{\rU(n)}\dif\mu_{\mathrm{Haar}}(\bsU)
\dif\mu_{\mathrm{Haar}}(\bsV)e^{\mathrm{i}\Tr{\bsX\bsU\widehat{\bsa}\bsU^\dagger}}e^{\mathrm{i}\Tr{\bsX\bsV\widehat{\bsb}\bsV^\dagger}}.
\end{eqnarray*}
We will consider the following question, i.e., the probability
density function of diagonals of sum of two Hermitian matrices:
\begin{eqnarray}
q(\bsC^\diag|\bsa,\bsb) = \int [\dif \bsC^{\mathrm{off}}]
p(\bsC|\bsa,\bsb).
\end{eqnarray}
Employ  the following identity:   $\bsX=\bsX^\diag\oplus
\bsX^{\mathrm{off}}$, $\bsC=\bsC^\diag\oplus \bsC^{\mathrm{off}}$,
and
\begin{eqnarray}
\delta\Pa{\bsX^{\mathrm{off}}} = \frac1{\pi^{n(n-1)}}\int [\dif
\bsC^{\mathrm{off}}]e^{-\mathrm{i}\Tr{\bsX^{\mathrm{off}}
\bsC^{\mathrm{off}}}}.
\end{eqnarray}
It follows that
\begin{eqnarray*}
&&q(\bsC^\diag|\bsa,\bsb) = \int [\dif \bsC^{\mathrm{off}}]
p(\bsC|\bsa,\bsb)\\
&& =\frac1{2^n\pi^{n^2}}\int[\dif \bsX]\int [\dif
\bsC^{\mathrm{off}}]e^{-\mathrm{i}\Tr{\bsX\bsC}}\int_{\rU(n)}\int_{\rU(n)}\dif\mu_{\mathrm{Haar}}(\bsU)\dif\mu_{\mathrm{Haar}}(\bsV)e^{\mathrm{i}\Tr{\bsX\bsU\widehat{\bsa}\bsU^\dagger}}
e^{\mathrm{i}\Tr{\bsX\bsV\widehat{\bsb}\bsV^\dagger}}\\
&&=\frac1{2^n\pi^{n^2}}\int[\dif \bsX]e^{-\mathrm{i}\Tr{\bsX^\diag
\bsC^\diag}}\int [\dif
\bsC^{\mathrm{off}}]e^{-\mathrm{i}\Tr{\bsX^{\mathrm{off}}
\bsC^{\mathrm{off}}}}\\
&&~~~~~~\times\int_{\rU(n)}\int_{\rU(n)}\dif\mu_{\mathrm{Haar}}(\bsU)
\dif\mu_{\mathrm{Haar}}(\bsV)e^{\mathrm{i}\Tr{\bsX\bsU\widehat{\bsa}\bsU^\dagger}}
e^{\mathrm{i}\Tr{\bsX\bsV\widehat{\bsb}\bsV^\dagger}}.
\end{eqnarray*}
Thus,
\begin{eqnarray*}
q(\bsC^\diag|\mathbf{a},\mathbf{b}) &=&\frac1{(2\pi)^n}\int[\dif
\bsX]e^{-\mathrm{i}\Tr{\bsX^\diag \bsC^\diag}}\delta\Pa{\bsX^{\mathrm{off}}}\\
&&\times\int_{\rU(n)}\int_{\rU(n)}\dif\mu_{\mathrm{Haar}}(\bsU)\dif\mu_{\mathrm{Haar}}(\bsV)
e^{\mathrm{i}\Tr{\bsX\bsU\widehat{\bsa}\bsU^\dagger}}e^{\mathrm{i}\Tr{\bsX\bsV\widehat{\bsb}\bsV^\dagger}}\\
&=&\frac1{(2\pi)^n}\int[\dif
\bsX^\diag]e^{-\mathrm{i}\Tr{\bsX^\diag \bsC^\diag}}\int_{\rU(n)}\int_{\rU(n)}\dif\mu_{\mathrm{Haar}}(\bsU)
\dif\mu_{\mathrm{Haar}}(\bsV)\\
&&\times\int[\dif
\bsX^{\mathrm{off}}]\delta\Pa{\bsX^{\mathrm{off}}}e^{\mathrm{i}\Tr{\bsX\bsU\widehat{\bsa}\bsU^\dagger}}e^{\mathrm{i}\Tr{\bsX\bsV\widehat{\bsb}\bsV^\dagger}}.
\end{eqnarray*}
Therefore,
\begin{eqnarray*}
q(\bsC^\diag|\bsa,\bsb)&=&\frac1{(2\pi)^n}\int[\dif
\bsX^\diag]e^{-\mathrm{i}\Tr{\bsX^\diag
\bsC^\diag}}\\
&&\times\int_{\rU(n)}\int_{\rU(n)}\dif\mu_{\mathrm{Haar}}(\bsU)\dif\mu_{\mathrm{Haar}}(\bsV)
e^{\mathrm{i}\Tr{\bsX^\diag
\bsU\widehat{\bsa}\bsU^\dagger}}e^{\mathrm{i}\Tr{\bsX^\diag
\bsV\widehat{\bsb}\bsV^\dagger}}
\end{eqnarray*}
and
\begin{eqnarray*}
q(\bsC^\diag|\bsa,\bsb)&=&\frac1{(2\pi)^n}\int[\dif
\bsX^\diag]e^{-\mathrm{i}\Tr{\bsX^\diag \bsC^\diag}}\Pa{\prod^n_{k=1}\Gamma(k)
\frac{\det\Pa{e^{\mathrm{i}\bsX_{ii}a_j}}}{\Delta(\mathrm{i}\bsX^\diag)\Delta(\bsa)}}
\Pa{\prod^n_{k=1}\Gamma(k)\frac{\det\Pa{e^{\mathrm{i}\bsX_{ii}b_j}}}{\Delta(\mathrm{i}\bsX^\diag)\Delta(\bsb)}}\\
&=&\frac1{(2\pi)^n\mathrm{i}^{2\binom{n}{2}}}\frac{\Pa{\prod^n_{k=1}\Gamma(k)}^2}{\Delta(\bsa)\Delta(\bsb)}\int\frac{[\dif
\bsX^\diag]}{\Delta(\bsX^\diag)^2}e^{-\mathrm{i}\Tr{\bsX^\diag
\bsC^\diag}}\det\Pa{e^{\mathrm{i}\bsX_{ii}a_j}}\det\Pa{e^{\mathrm{i}\bsX_{ii}b_j}}.
\end{eqnarray*}
This completes the proof.
\end{proof}

\begin{lem}
For given real numbers $x_j$ and $\lambda_j$ where $j=1,\ldots,n$,
denote $\bar{x} = \frac1n\sum^n_{k=1}x_k$, then have the following
identities:
\begin{eqnarray} \label{eqn:CorProof_Det}
\det\Pa{e^{\mathrm{i}x_i\lambda_j}}=e^{\mathrm{i}\bar{x}\sum^n_{k=1}\lambda_k}\sum_{\sigma\in
S_n} \sign(\sigma)
\prod^{n-1}_{k=1}\exp\Br{\mathrm{i}(x_k-x_{k+1})\Pa{\sum^k_{j=1}\lambda_{\sigma(j)}
-\frac kn\sum^n_{j=1}\lambda_j}}
\end{eqnarray}
and
\begin{eqnarray} \label{eqn:CorProof_Product}
\prod^n_{k=1}e^{-\mathrm{i}x_kc_k}=
\exp\Br{-\mathrm{i}\bar{x}\sum^n_{k=1}c_k}
\exp\Br{-\mathrm{i}\sum^{n-1}_{k=1}(x_k-x_{k+1})\Pa{\sum^k_{j=1}c_j
- \frac kn\sum^k_{j=1}c_j}}.
\end{eqnarray}
\end{lem}

\begin{proof}
With the definition of $\bar x$, we have
\begin{eqnarray*}
\det\Pa{e^{\mathrm{i}x_i\lambda_j}}=e^{\mathrm{i}\bar{x}\sum^n_{k=1}\lambda_k}\det\Pa{e^{\mathrm{i}(x_i-\bar{x})\lambda_j}}.
\end{eqnarray*}
Next we recall the Abel's identity as follows. In fact, we have
\begin{eqnarray*}
\sum^n_{k=1} x_ky_k = y_n\sum^n_{k=1}x_k +
\sum^{n-1}_{k=1}(y_k-y_{k+1})\Pa{\sum^k_{j=1}x_j}.
\end{eqnarray*}
Now we expand the determinant
$\det\Pa{e^{\mathrm{i}(x_i-\bar{x})\lambda_j}}$ as below:
\begin{eqnarray*}
\det\Pa{e^{\mathrm{i}(x_i-\bar{x})\lambda_j}} &=& \sum_{\sigma\in
S_n}
\sign(\sigma)\prod^n_{k=1}e^{\mathrm{i}(x_k-\bar{x})\lambda_{\sigma(k)}}
=\sum_{\sigma\in S_n}
\sign(\sigma)\exp\Pa{\mathrm{i}\sum^n_{k=1}(x_k-\bar{x})\lambda_{\sigma(k)}},
\end{eqnarray*}
where $\sign(\sigma)=\pm1$ for odd $(-1)$ or even $(+1)$
permutation. By using Abel's identity, we have
\begin{eqnarray*}
\sum^n_{k=1}\lambda_{\sigma(k)}(x_k-\bar{x}) &=&
(x_n-\bar{x})\sum^n_{k=1}\lambda_{\sigma(k)} +
\sum^{n-1}_{k=1}(x_k-x_{k+1})\Pa{\sum^k_{j=1}\lambda_{\sigma(j)}}\\
&=&
\sum^{n-1}_{k=1}(x_k-x_{k+1})\Pa{\sum^k_{j=1}\lambda_{\sigma(j)}} -
(\bar{x}-x_n)\sum^n_{j=1}\lambda_j.
\end{eqnarray*}
Again, by using Abel's identity, we have
\begin{eqnarray*}
n\bar{x} = \sum^n_{k=1}1\cdot x_k = x_n\sum^n_{k=1}1 +
\sum^{n-1}_{k=1} (x_k-x_{k+1})\Pa{\sum^k_{j=1}1} = nx_n+
\sum^{n-1}_{k=1} k(x_k-x_{k+1}).
\end{eqnarray*}
That is,
\begin{eqnarray*}
\bar{x} - x_n= \sum^{n-1}_{k=1} \frac kn(x_k-x_{k+1}).
\end{eqnarray*}
Finally, we have
\begin{eqnarray*}
\sum^n_{k=1}\lambda_{\sigma(k)}(x_k-\bar{x}) &=&
\sum^{n-1}_{k=1}(x_k-x_{k+1})\Pa{\sum^k_{j=1}\lambda_{\sigma(j)}} -
\Pa{\sum^{n-1}_{k=1} \frac kn(x_k-x_{k+1})}\sum^n_{j=1}\lambda_j\\
&=&\sum^{n-1}_{k=1}(x_k-x_{k+1})\Pa{\sum^k_{j=1}\lambda_{\sigma(j)}
-\frac kn\sum^n_{j=1}\lambda_j}.
\end{eqnarray*}
That is,
\begin{eqnarray*}
\sum^n_{k=1}\lambda_{\sigma(k)}(x_k-\bar{x}) =
\sum^{n-1}_{k=1}(x_k-x_{k+1})\Pa{\sum^k_{j=1}\lambda_{\sigma(j)}
-\frac kn\sum^n_{j=1}\lambda_j}.
\end{eqnarray*}
From the above discussion, we see that
\begin{eqnarray*}
\det\Pa{e^{\mathrm{i}x_i\lambda_j}}=e^{\mathrm{i}\bar{x}\sum^n_{k=1}\lambda_k}\sum_{\sigma\in
S_n} \sign(\sigma)
\prod^{n-1}_{k=1}\exp\Br{\mathrm{i}(x_k-x_{k+1})\Pa{\sum^k_{j=1}\lambda_{\sigma(j)}
-\frac kn\sum^n_{j=1}\lambda_j}}.
\end{eqnarray*}
Note that
$\prod^n_{k=1}e^{-\mathrm{i}x_kc_k}=\exp\Br{-\mathrm{i}\sum^n_{k=1}c_kx_k}$,
where
$$
\sum^n_{k=1}c_kx_k = x_n\sum^n_{k=1}c_k +
\sum^{n-1}_{k=1}(x_k-x_{k+1})\Pa{\sum^k_{j=1}c_j}.
$$
Hence, via $x_n= \bar{x} - \sum^{n-1}_{k=1} \frac kn(x_k-x_{k+1})$,
it follows that
\begin{eqnarray}\label{eq:ip-formula}
\sum^n_{k=1}c_kx_k = \bar{x}\sum^n_{k=1}c_k +
\sum^{n-1}_{k=1}(x_k-x_{k+1})\Pa{\sum^k_{j=1}c_j - \frac
kn\sum^k_{j=1}c_j},
\end{eqnarray}
and we have the formula \eqref{eqn:CorProof_Product}.
This completes the proof.
\end{proof}
In the following we derive a further simplified expression for the
analytical formula in Theorem~\ref{thrm:main}. By
Theorem~\ref{thrm:main}, we see that it suffices to calculate the
following integral
\begin{eqnarray}
I(\bsa,\bsb:\bsC^\diag)&:=&\int_{\real^n}[\dif\bsx]\det\Pa{e^{\mathrm{i}x_ia_j}}\det\Pa{e^{\mathrm{i}x_ib_j}}\Delta(\bsx)^{-2}\prod^n_{i=1}e^{-\mathrm{i}x_i\bsC_{ii}}.
\end{eqnarray}
\begin{cor}\label{cor:main}
The joint pdf $q(\bsC^\diag|\bsa,\bsb)$ of the diagonal part
$\bsC^\diag$ of the sum $\bsC=\bsA+\bsB$ is given by
\begin{eqnarray}\label{eq:q(Cdiag|a,b)reformula}
q(\bsC^\diag|\bsa,\bsb)
=\frac{M}{\Delta(\bsa)\Delta(\bsb)}I(\bsa,\bsb:\bsC^\diag),
\end{eqnarray}
where $M$ is from \eqref{eq:constM} and
\begin{eqnarray}\label{eq:I(a,b:Cdiag)inCor2.3}
I(\bsa,\bsb:\bsC^\diag)
=2\pi\delta\Pa{\sum^n_{j=1}(a_j+b_j-\bsC_{jj})}\sum_{\sigma,\tau\in
S_n}\sign(\sigma\tau)\int_{\real^{n-1}}\frac{[\dif\bsu]}{\widetilde\Delta(\bsu)^2}
\prod^{n-1}_{k=1}e^{\mathrm{i}u_kB_k(\sigma,\tau)}
\end{eqnarray}
with
\begin{eqnarray} \label{eqn:BksigmainCor}
B_k(\sigma,\tau):=\sum^k_{j=1}(a_{\sigma(j)}+b_{\tau(j)}-\bsC_{jj})
-\frac kn\sum^n_{j=1}(a_j+b_j-\bsC_{jj}).
\end{eqnarray}
\end{cor}

\begin{proof}
Note that \eqref{eq:q(Cdiag|a,b)reformula} is the reformulation of
Theorem \ref{thrm:main}. So we focus on the expression
$I(\bsa,\bsb:\bsC^\diag)$. From the formulae
\eqref{eqn:CorProof_Det} and \eqref{eqn:CorProof_Product} we have
\begin{eqnarray*}
\det\Pa{e^{\mathrm{i}x_ia_j}}&=&\exp\Br{\mathrm{i}\bar{x}\sum^n_{k=1}a_k}\sum_{\sigma\in
S_n} \sign(\sigma)
\prod^{n-1}_{k=1}\exp\Br{\mathrm{i}(x_k-x_{k+1})\Pa{\sum^k_{j=1}a_{\sigma(j)}
-\frac kn\sum^n_{j=1}a_j}}, \\
\det\Pa{e^{\mathrm{i}x_ib_j}}&=&\exp\Br{\mathrm{i}\bar{x}\sum^n_{k=1}b_k}\sum_{\tau\in
S_n} \sign(\tau)
\prod^{n-1}_{k=1}\exp\Br{\mathrm{i}(x_k-x_{k+1})\Pa{\sum^k_{j=1}b_{\tau(j)}
-\frac kn\sum^n_{j=1}b_j}}.
\end{eqnarray*}
and
\begin{eqnarray*}
\prod^n_{k=1}e^{-\mathrm{i}x_k\bsC_{kk}}=
\exp\Br{-\mathrm{i}\bar{x}\sum^n_{k=1}\bsC_{kk}}
\exp\Br{\mathrm{i}\sum^{n-1}_{k=1}(x_k-x_{k+1})\Pa{-\sum^k_{j=1}\bsC_{jj}
+ \frac kn\sum^k_{j=1}\bsC_{jj}}}.
\end{eqnarray*}
And it follows that
\begin{eqnarray*}
&&\det\Pa{e^{\mathrm{i}x_ia_j}}\det\Pa{e^{\mathrm{i}x_ib_j}}\prod^n_{k=1}e^{-\mathrm{i}x_k\bsC_{kk}}\\
&&=
\exp\Br{\mathrm{i}\bar{x}\sum^n_{j=1}(a_j+b_j-\bsC_{jj})}\sum_{\sigma,\tau\in
S_n} \sign(\sigma\tau)
\prod^{n-1}_{k=1}\exp\Br{\mathrm{i}(x_k-x_{k+1})B_k(\sigma,\tau)},
\end{eqnarray*}
where  $B_k(\sigma,\tau)$ is defined in \eqref{eqn:BksigmainCor}. We
next perform the  change of variables: $(x_1,\ldots,x_n)\to (\bar x,
u_1,\ldots,u_{n-1})$, where $u_k = x_k-x_{k+1}$. The Jacobian of
this transformation is given by
\begin{eqnarray}
J = \abs{\frac{\partial(\bar x,
u_1,\ldots,u_{n-1})}{\partial(x_1,\ldots,x_n)}}
=\abs{\begin{array}{ccccc}
                                                        \frac1n & \frac1n & \cdots & \frac1n & \frac1n \\
                                                        1 & -1 & \cdots & 0 & 0 \\
                                                        0 & 1 & \ddots & 0 & 0 \\
                                                        \vdots & \vdots & \ddots & -1 & 0 \\
                                                        0 & 0 & \cdots & 1 &
                                                        -1
                                                      \end{array}
}= (-1)^{n-1}.
\end{eqnarray}
Then $\dif x_1\wedge\cdots\wedge\dif
x_n=(-1)^{n-1}\dif\bar{x}\wedge\dif u_1\wedge\cdots\wedge\dif
u_{n-1}$, that is, this transformation is volume-preserving, $[\dif
\bsx]=\dif\bar{x}[\dif \bsu]$, where $[\dif
\bsu]=\prod^{n-1}_{j=1}\dif u_j$. Now we have already known
\cite{Hoskins2009} that
$$
\delta(s) = \frac1{2\pi}\int_\real e^{\mathrm{i}st}\dif t.
$$
This indicates that the first factor is given by
$$
\int_\real e^{\mathrm{i}\bar{x}\sum^n_{k=1}(a_k+b_k-c_k)}\dif\bar x
= 2\pi\delta\Pa{\sum^n_{k=1}(a_k+b_k-c_k)}.
$$
Finally, we have
\begin{eqnarray*}
I(\bsa,\bsb:\bsC^\diag) &=&
\int_{\real^n}\frac{[\dif\bsx]}{\Delta(\bsx)^2}\det\Pa{e^{\mathrm{i}x_ia_j}}\det\Pa{e^{\mathrm{i}x_ib_j}}\prod^n_{k=1}e^{-\mathrm{i}x_k\bsC_{kk}}\\
&=&\int_\real e^{\mathrm{i}\bar{x}\sum^n_{k=1}(a_k+b_k-c_k)}\dif\bar
x\sum_{\sigma,\tau\in
S_n}\sign(\sigma\tau)\int_{\real^{n-1}}\frac{[\dif\bsu]}{\widetilde\Delta(\bsu)^2}
\prod^{n-1}_{k=1}e^{\mathrm{i}u_kB_k(\sigma,\tau)}\\
&=&2\pi\delta\Pa{\sum^n_{j=1}(a_j+b_j-\bsC_{jj})}\sum_{\sigma,\tau\in
S_n}\sign(\sigma\tau)\int_{\real^{n-1}}\frac{[\dif\bsu]}{\widetilde\Delta(\bsu)^2}
\prod^{n-1}_{k=1}e^{\mathrm{i}u_kB_k(\sigma,\tau)},
\end{eqnarray*}
where
$$
\widetilde \Delta(\bsu) := \prod_{1\leqslant i\leqslant j-1\leqslant
n-1}(u_i+u_{i+1}+\cdots +u_{j-1}).
$$
This completes the proof.
\end{proof}

\begin{exam}
For the sum of two $2\times 2$ random Hermitian matrices, we can
derive concise expressions for the pdfs of diagonals of this sum.
For $n=2$, for $\bsa=(a_1,a_2)$ with $a_1\geqslant a_2$ and
$\bsb=(b_1,b_2)$ with $b_1\geqslant b_2$,  the formula
\eqref{eq:I(a,b:Cdiag)inCor2.3} can be simplifed as follows.
\begin{eqnarray}\label{eq:I(a,b:Cdiag)n=2}
I(\bsa,\bsb:\bsC^\diag)
=2\pi\delta\Pa{\sum^2_{j=1}(a_j+b_j-\bsC_{jj})}\sum_{\sigma,\tau\in
S_2}\sign(\sigma\tau)\int_{\real}\frac{\dif u}{u^2}
e^{\mathrm{i}uB(\sigma,\tau)},
\end{eqnarray}
where
\begin{eqnarray*}
B(\sigma,\tau):=a_{\sigma(1)}+b_{\tau(1)}-\bsC_{11} -\frac
12\sum^2_{j=1}(a_j+b_j-\bsC_{jj}).
\end{eqnarray*}
Next, we calculate the integral:
$$
\int^{+\infty}_{-\infty}\frac{e^{\mathrm{i}uB(\sigma,\tau)}}{u^2}
\dif u.
$$
We apply the following formula for Fourier transform
\cite{Hoskins2009}
\begin{eqnarray}
\int^{+\infty}_{-\infty}\frac{e^{-\mathrm{i} w t}}{t^n}\dif t =
-\mathrm{i}\pi\frac{(-\mathrm{i}w)^{n-1}}{\Gamma(n)}\mathrm{sgn}(w),
\end{eqnarray}
where
$$
\mathrm{sgn}(w) = \begin{cases}1,&\text{if }w>0,\\0,&\text{if }
w=0,\\-1,&\text{if }w<0.\end{cases}
$$
We can see that
\begin{eqnarray}
\int^{+\infty}_{-\infty}\frac{e^{-\mathrm{i}\nu t}}{t^2}\dif t =
-\pi\abs{\nu},
\end{eqnarray}
and thus,
\begin{eqnarray}
\int^{+\infty}_{-\infty}\frac{e^{\mathrm{i}uB(\sigma)}}{u^2} \dif u
= -\pi \abs{B(\sigma,\tau)}.
\end{eqnarray}
From the above, we see that
\begin{eqnarray}
\sum_{\sigma,\tau\in S_2}\sign(\sigma\tau)\int_{\real}\frac{\dif
u}{u^2} e^{\mathrm{i}uB(\sigma,\tau)} = -\pi \sum_{\sigma,\tau\in
S_2}\sign(\sigma\tau)\abs{B(\sigma,\tau)}.
\end{eqnarray}
Define $\alpha_{12}:=a_1-a_2\geqslant0$, $\beta_{12}
:=b_1-b_2\geqslant0$, and $\tilde\gamma_{12}:=\bsC_{11}-\bsC_{22}$.
Then \eqref{eq:I(a,b:Cdiag)n=2} is rewritten as
\begin{eqnarray*}
I(\bsa,\bsb:\bsC^\diag) &=&
2\pi\delta\Pa{\sum^2_{j=1}(a_j+b_j-\bsC_{jj})}\\
&&\times\Pa{-\frac\pi2\sum_{\sigma,\tau\in
S_2}\sign(\sigma\tau)\abs{\sign(\sigma)\alpha_{12} +
\sign(\tau)\beta_{12} - \tilde \gamma_{12}}}.
\end{eqnarray*}
That is,
\begin{eqnarray*}
I(\bsa,\bsb:\bsC^\diag)
&=&\pi^2\delta\Pa{\sum^2_{j=1}(a_j+b_j-\bsC_{jj})} (\abs{\alpha_{12}
- \beta_{12} - \tilde \gamma_{12}} + \abs{\alpha_{12} - \beta_{12} +
\tilde \gamma_{12}}\\
&& - \abs{\alpha_{12} + \beta_{12} - \tilde \gamma_{12}} -
\abs{\alpha_{12} + \beta_{12} + \tilde \gamma_{12}}).
\end{eqnarray*}
Therefore
\begin{eqnarray} \label{eqn:q(c|ab)4lowerdim}
q(\bsC^\diag|\bsa,\bsb)
&=&\frac1{4\alpha_{12}\beta_{12}}\delta\Pa{\sum^2_{j=1}(a_j+b_j-\bsC_{jj})}(\abs{\alpha_{12}
+ \beta_{12} -
\tilde \gamma_{12}} + \abs{\alpha_{12} + \beta_{12} + \tilde \gamma_{12}}\notag\\
&&- \abs{\alpha_{12} - \beta_{12} - \tilde \gamma_{12}} -
\abs{\alpha_{12} - \beta_{12} + \tilde \gamma_{12}}).
\end{eqnarray}

\end{exam}

\section{The pdf of eigenvalues of the sum of random Hermitian matrices}\label{sect:der-pin}

The derivative principle is formally put forward in
\cite{Christandl2014}. The authors of \cite{Christandl2014} derived
this result in the abstract level, i.e., in the regime of Lie
algebra, and they used this result to obtain the distribution of
eigenvalues of random marginals of a multipartite random pure state.
Later, Mej\'{i}a, Zapata, and Botero rederived this result in Random
Matrix Theory (RMT) \cite{Mejia2017}, and they used this result to
study the difference between two random mixed quantum states. The
following version of the derivative principle is from
\cite{Mejia2017}.

\begin{prop}[The derivative principle]\label{prop:der-princ}
Let $\bsZ$ be an $n\times n$ random matrix drawn from a unitarily
invariant random matrix ensemble, $p_{\bsZ}$ the joint eigenvalue
distribution for $\bsZ$ and $q_{\bsZ}$ the joint distribution of the
diagonal elements of $\bsZ$. Then
\begin{eqnarray}\label{eq:derivative-principle}
p_{\bsZ}(\lambda) =
\frac1{\prod^n_{k=1}k!}(-1)^{\binom{n}{2}}\Delta(\lambda)\Delta(\partial_\lambda)
q_{\bsZ}(\lambda),
\end{eqnarray}
where $\Delta(\lambda) = \prod_{i<j}(\lambda_i-\lambda_j)$ is the
Vandermonde determinant and $\Delta(\partial_\lambda)$ the
differential operator
$\prod_{i<j}\Pa{\frac{\partial}{\partial_{\lambda_i}} -
\frac{\partial}{\partial_{\lambda_j}}}$.
\end{prop}

\subsection{A new derivation of the pdf by the derivative principle}

With the derivative principle, we can relate the pdf of the
eigenvalues of the sum of two random Hermitian matrices with given
spectra to the pdf of diagonals of this sum.

In the following, we will use the derivative principle to rederive
the pdf of eigenvalues of sum of two random Hermitian matrices with
given spectra.
\begin{thrm}\label{thrm:CorollaryDP}
Assume that two random Hermitian matrices $\bsA$ and $\bsB$ chosen
uniformly on the unitary orbits $\cU(\bsa)$ and $\cU(\bsb)$,
respectively, then the joint pdf $p(\bsc|\bsa,\bsb)$ of the eigenvalues
$\bsc$ of the sum $\bsC=\bsA+\bsB$ is given by derivative principle
\eqref{eq:derivative-principle}
\begin{eqnarray}
p(\bsc|\bsa,\bsb) =
\frac1{\prod^n_{k=1}k!}(-1)^{\binom{n}{2}}\Delta(\bsc)\Delta(\partial_{\bsc})
q(\bsc|\bsa,\bsb).
\end{eqnarray}
where $q(\bsc|\bsa,\bsb)$ is from \eqref{eq:q(Cdiag|a,b)reformula}
by replacing $\bsC^\diag$ as $\bsc$. Moreover,
\begin{eqnarray} \label{eqn:q(c|ab)inThm}
p(\bsc|\bsa,\bsb)
=\frac{\prod^n_{k=1}\Gamma(k)}{(2\pi)^nn!\mathrm{i}^{\binom{n}{2}}}
\frac{\Delta(\bsc)}{\Delta(\bsa)\Delta(\bsb)}\int_{\real^n}\frac{[\dif
\bsx]}{\Delta(\bsx)}
\det\Pa{e^{\mathrm{i}x_ia_j}}\det\Pa{e^{\mathrm{i}x_ib_j}}\prod^n_{k=1}e^{-\mathrm{i}x_kc_k}.
\end{eqnarray}
\end{thrm}

\begin{proof}
Note that
\begin{eqnarray*}
\Pa{\frac{\partial}{\partial_{c_i}} -
\frac{\partial}{\partial_{c_j}}}q(\bsc|\bsa,\bsb) =
\frac{M}{\Delta(\bsa)\Delta(\bsb)} \int\det\Pa{e^{\mathrm{i}x_ia_j}}
\det\Pa{e^{\mathrm{i}x_ib_j}}\Delta(\bsx)^{-2}\Pa{\frac{\partial}{\partial_{c_i}}
-
\frac{\partial}{\partial_{c_j}}}\prod^n_{k=1}e^{-\mathrm{i}x_kc_k}\dif
x_k,
\end{eqnarray*}
where
\begin{eqnarray*}
\Pa{\frac{\partial}{\partial_{c_i}} -
\frac{\partial}{\partial_{c_j}}}\prod^n_{k=1}e^{-\mathrm{i}x_kc_k} =
(-\mathrm{i})(x_i-x_j)\prod^n_{k=1}e^{-\mathrm{i}x_kc_k}.
\end{eqnarray*}
Thus,
\begin{eqnarray*}
\Delta(\partial_{\bsc})\prod^n_{k=1}e^{-\mathrm{i}x_kc_k} &=&
\prod_{i<j}(-\mathrm{i})(x_i-x_j)\prod^n_{k=1}e^{-\mathrm{i}x_kc_k}\\
&=&\frac1{\mathrm{i}^{\binom{n}{2}}}\Delta(\bsx)\prod^n_{k=1}e^{-\mathrm{i}x_kc_k}.
\end{eqnarray*}
Therefore,
\begin{eqnarray*}
\Delta(\partial_{\bsc})q(\bsc|\bsa,\bsb) =
\frac1{\mathrm{i}^{\binom{n}{2}}}\frac{M}{\Delta(\bsa)\Delta(\bsb)}
\int\frac{[\dif \bsx]}{\Delta(\bsx)}\det\Pa{e^{\mathrm{i}x_ia_j}}
\det\Pa{e^{\mathrm{i}x_ib_j}}\prod^n_{k=1}e^{-\mathrm{i}x_kc_k}.
\end{eqnarray*}
 Substituting it into the right hand side of \eqref{eq:derivative-principle}, we have
\begin{eqnarray*}
&&\frac1{\prod^n_{k=1}k!}(-1)^{\binom{n}{2}}\Delta(\bsc)\Delta(\partial_{\bsc})q(\bsc|\bsa,\bsb)
\\
&&=\frac{(-1)^{\binom{n}{2}}}{n!\prod^n_{k=1}\Gamma(k)}
\frac{M}{\mathrm{i}^{\binom{n}{2}}}\frac{\Delta(\bsc)}{\Delta(\bsa)\Delta(\bsb)}
\int\frac{[\dif \bsx]}{\Delta(\bsx)}\det\Pa{e^{\mathrm{i}x_ia_j}}
\det\Pa{e^{\mathrm{i}x_ib_j}}\prod^n_{k=1}e^{-\mathrm{i}x_kc_k}\\
&&=
\frac{\prod^n_{k=1}\Gamma(k)}{(2\pi)^nn!\mathrm{i}^{\binom{n}{2}}}\frac{\Delta(\bsc)}{\Delta(\bsa)\Delta(\bsb)}
\int\frac{[\dif \bsx]}{\Delta(\bsx)}\det\Pa{e^{\mathrm{i}x_ia_j}}
\det\Pa{e^{\mathrm{i}x_ib_j}}\prod^n_{k=1}e^{-\mathrm{i}x_kc_k}.
\end{eqnarray*}
This is exactly   the application of   derivative principle,  and it
yields the joint eigenvalue distribution.
\end{proof}

\begin{remark}
Denote by $\lambda(\bsX)$ the vector whose components consisting of
eigenvalues in a non-increasing order of Hermitian matrix $\bsX$.
Given two $2\times 2$ Hermitian matrices $\bsA$ and $\bsB$ with
$\lambda(\bsA)=(a_1,a_2):=\bsa$ and $\lambda(\bsB)=(b_1,b_2)=\bsb$.
We also denote $\bsC=\bsA+\bsB$ and $\lambda(\bsC)=(c_1,c_2):=\bsc$.
Denote $I=(\abs{\alpha_{12}-\beta_{12}}, \alpha_{12}+\beta_{12})$,
where $\alpha_{12}=a_1-a_2\geqslant0$ and
$\beta_{12}=b_1-b_2\geqslant0$. The pdf of $\bsC=\bsU\widehat
\bsa\bsU^\dagger+\bsV\widehat \bsb\bsV^\dagger$ is given as
\cite{JBZ2017}
\begin{eqnarray}\label{eq:sum-orbits}
p(\bsc|\bsa,\bsb) &=&
\frac{c_1-c_2}{2(a_1-a_2)(b_1-b_2)}(1_{I}(c_1-c_2)-1_{-I}(c_1-c_2))\notag\\
&&\times\delta\Pa{c_1+c_2-a_1-a_2-b_1-b_2},
\end{eqnarray}
which can be also derived from \eqref{eqn:q(c|ab)4lowerdim} and the
derivative principle.
\end{remark}

\subsection{The pdf of eigenvalues of the sum of random Hermitian matrices from GUE ensemble}\label{sect:sum-of-tworan}

Recall that the standard complex normal random variable or standard
complex Gaussian random variable is a complex random variable $z$
whose real and imaginary parts are independent normally distributed
random variables with mean zero and variance $\frac12$. We use the
notation $z\sim N_\complex(0,1)=N(0,\frac12)+\sqrt{-1}N(0,\frac12)$
to denote the fact that $z$ is the standard complex normal random
variable.

The so-called $\rG\rU\rE(n)$ ensemble is the class of complex
Hermitian random matrices $\bsA=(a_{ij})_{n\times n}$, generated in
the following way:
\begin{eqnarray}\label{eq:half-sum}
\bsA=\frac12(\bsZ+\bsZ^\dagger),
\end{eqnarray}
where $\bsZ=(z_{ij})_{n\times n}$ is the standard complex Gaussian
random matrix, i.e., $z_{ij}\sim N_\complex(0,1)$ are i.i.d. for all
$i,j$. From this, we see that
$$
a_{ij} = \frac12(z_{ij}+\bar z_{ji}).
$$
If $i=j$, then $a_{ii}=\re(z_{ii})\sim N(0,\frac12)$; if $i\neq j$,
then $a_{ij}\sim N_\complex(0,\frac12)$. Moreover, the density
functions of $a_{ij}$ are given by
\begin{eqnarray}
p(a_{ij}) = \begin{cases} \frac1{\sqrt{\pi}}e^{-a^2_{ii}},&\text{if
}i=j; \\
p(a_{ij}) = \frac2\pi e^{-2\abs{a_{ij}}^2},&\text{if }i\neq j.
\end{cases}
\end{eqnarray}
Therefore the pdf of a random matrix $\bsA\in\rG\rU\rE(n)$ is given
by the following:
\begin{eqnarray}\label{eq:gue}
p(\bsA) =\frac{2^{\binom{n}{2}}}{\pi^{\frac{n^2}2}}
\exp\Pa{-\Tr{\bsA^2}}.
\end{eqnarray}
As already known \cite{Andrews1999},
\begin{eqnarray}
\int_{\real^n} \Delta(\bsx)^2\exp\Pa{-\inner{\bsx}{\bsx}}[\dif
\bsx]=(2\pi)^{\frac n2}2^{-\frac{n^2}2}\prod^n_{k=1}k! .
\end{eqnarray}

The pdf of eigenvalues $\bsa=(a_1,a_2,\ldots,a_n)$ of
$\bsA=\bsU\widehat \bsa\bsU^\dagger$ is given by
\begin{eqnarray}
\frac{2^{\binom{n}{2}}}{\pi^{\frac
n2}\prod^n_{k=1}k!}\Delta(\bsa)^2\exp\Pa{-\inner{\bsa}{\bsa}},\quad
\bsa\in\real^n.
\end{eqnarray}
Here $\Inner{\cdot}{\cdot}$ is the Euclid inner product.

\begin{prop}\label{th:sum-of-Kgue}
The pdf of $\bss=(s_1,\ldots,s_n)$ eigenvalues of the sum
$\bsS=\sum^K_{k=1}\bsA_k$, where all $\bsA_k\in\rG\rU\rE(n)$ are
i.i.d., is given by
\begin{eqnarray}
p_{\bsS,K}(\bss) = \frac{\Pa{\frac2K}^{\binom{n}{2}}}{(K\pi)^{\frac
n2}\prod^n_{j=1}j!}\Delta(\bss)^2\exp\Pa{-\frac1K\inner{\bss}{\bss}},\quad\bss\in\real^n.
\end{eqnarray}
\end{prop}

\begin{proof}
As an illustration, we show the corresponding result for $K=2$ and
$\bsS=\bsA_1+\bsA_2$. The proof for an arbitrary positive integer
$K$ goes similarly. Firstly we work out the pdf of the diagonal part
of $\bsW$. In fact,
$$
\bsS^\diag = \diag(a^{(1)}_{11}+a^{(2)}_{11},
a^{(1)}_{22}+a^{(2)}_{22},\ldots,a^{(1)}_{nn}+a^{(2)}_{nn}),
$$
where $a^{(1)}_{ii},a^{(2)}_{jj}\sim N(0,\frac12)$ are i.i.d.
Gaussian random variables. Then the pdf $q\Pa{\bsS^\diag}$ of
$\bsS^\diag$ is given
$$
q\Pa{\bsS^\diag} = \prod^n_{i=1}q(S_{ii}),
$$
where
$$
q(S_{ii})=(\varphi\star \varphi)(S_{ii}) =
\frac1{\sqrt{2\pi}}\exp\Pa{-\frac12S^2_{ii}}
$$
for $i=1,\cdots, n$. Here
$\varphi(x) = \frac1{\sqrt{\pi}}e^{-x^2}$ and $\star$ means the
convolution product. Thus
\begin{eqnarray}
q\Pa{\bsS^\diag} =
\prod^n_{i=1}\frac1{\sqrt{2\pi}}\exp\Pa{-\frac12S^2_{ii}} =
\Pa{\frac1{2\pi}}^{\frac n2}\exp\Pa{-\frac12\sum^n_{i=1}S^2_{ii}}.
\end{eqnarray}
By using Proposition~\ref{prop:der-princ}, we have
\begin{eqnarray}
p_{\bsS}(\bss) &=& p_{\bsA_1+\bsA_2}(\bss)
=\frac1{\prod^n_{k=1}k!}(-1)^{\binom{n}{2}}\Delta(\bss)\Delta(\partial_{\bss})q(\bss)\\
&=&\frac1{(2\pi)^{\frac
n2}\prod^n_{k=1}k!}\Delta(\bss)^2\exp\Pa{-\frac12\inner{\bss}{\bss}}.
\end{eqnarray}
This completes the proof.
\end{proof}
Note that this same result can be derived as a direct corollary of
the \eqref{eq:half-sum} in the section. The above working amply
demonstrate the consistency.

\subsection{The connection with Golden-Thompson
inequality}\label{sect:GTinequality}

Let $\dif\mu(\bsH) =p(\bsH)[\dif \bsH]$, where $p(\bsH)$ the pdf
which is given by \eqref{eq:gue}. Denote
\begin{eqnarray}\label{eq:expect}
\mathbb{E}_{\bsH\in \rG\rU\rE(n)}\Br{f(\bsH)} =\int_{\rG\rU\rE(n)}
f(\bsH)\dif \mu(\bsH),
\end{eqnarray}
where $f(\bsH)$ is the functional calculus of the function $f$. By
the spectral decomposition of $\bsH$: $\bsH=\bsU\Lambda
\bsU^\dagger$, where $\Lambda=\diag(\lambda_1,\ldots,\lambda_n)$
with $+\infty>\lambda_1>\cdots>\lambda_n>-\infty$, and
$\bsU\in\rU(n)/\mathbb{T}$. Here $\mathbb{T}$ is the maximal torus
of $\rU(n)$, and for $f$ in \eqref{eq:expect} matrix valued
$f(\bsH)=\bsU f(\Lambda)\bsU^\dagger$, where
$f(\Lambda)=\diag(f(\lambda_1),\ldots,f(\lambda_n))$. Then we have
\cite{Zhang2018}
\begin{eqnarray}\label{eq:LV}
[\dif \bsH] = \frac{\vol(\rU(n))}{(2\pi)^n}\Delta(\lambda)^2[\dif
\Lambda]\dif\nu(\bsU),
\end{eqnarray}
where
\begin{eqnarray*}
\vol(\rU(n)) =
\frac{2^n\pi^{\binom{n+1}{2}}}{\prod^n_{k=1}\Gamma(k)}
\end{eqnarray*}
is the volume associated with the Haar measure on $\rU(n)$ and
$\dif\nu(\bsU)$ is the normalized Haar measure over $\rU(n)$ in the
sense that
\begin{eqnarray*}
\int_{\rU(n)} \dif\nu(\bsU) = 1.
\end{eqnarray*}
We get that
\begin{eqnarray}
\dif\mu(\bsH) = \frac{2^{\binom{n}{2}}}{\pi^{\frac
n2}\prod^n_{k=1}\Gamma(k)}
\Delta(\lambda)^2\exp\Pa{-\Tr{\Lambda^2}}[\dif\Lambda]\dif\nu(\bsU).
\end{eqnarray}
Moreover,
\begin{eqnarray*}
\mathbb{E}_{\bsH\in
\rG\rU\rE(n)}\Br{f(\bsH)}&=&\frac{2^{\binom{n}{2}}}{\pi^{\frac
n2}\prod^n_{k=1}\Gamma(k)}\int_{+\infty>\lambda_1>\cdots>\lambda_n>-\infty}
\Delta(\lambda)^2\exp\Pa{-\Tr{\Lambda^2}}[\dif\Lambda]\\
&&\times\int\dif\nu(\bsU)\bsU f(\Lambda)\bsU^\dagger \\
&=&\kappa_n(f)\cdot\I_n,
\end{eqnarray*}
where
\begin{eqnarray*}
\kappa_n(f)&\defeq& \frac{2^{\binom{n}{2}}}{n\pi^{\frac
n2}\prod^n_{k=1}\Gamma(k)}\int_{+\infty>\lambda_1>\cdots>\lambda_n>-\infty}\Tr{f(\Lambda)}
\Delta(\lambda)^2\exp\Pa{-\Tr{\Lambda^2}}[\dif\Lambda]\notag\\
&=&\frac1n\frac{2^{\binom{n}{2}}}{\pi^{\frac
n2}\prod^n_{k=1}k!}\int_{\real^n}\Tr{f(\Lambda)}
\Delta(\lambda)^2\exp\Pa{-\Tr{\Lambda^2}}[\dif\Lambda]\\
&=&\frac1n \int_{\real^n}\Tr{f(\Lambda)}p(\Lambda)[\dif\Lambda],
\end{eqnarray*}
and
$$
p(\Lambda) = \frac{2^{\binom{n}{2}}}{\pi^{\frac
n2}\prod^n_{k=1}k!}\Delta(\lambda)^2\exp\Pa{-\Tr{\Lambda^2}}.
$$
Note that in the above reasoning, we used the fact that
$$
\int_{\rU(n)}\dif\nu(\bsU)\bsU f(\Lambda) \bsU^\dagger =
\frac{\Tr{f(\Lambda)}}n\I_n.
$$

\begin{prop}[\cite{Mehta2004}]
Let
$$
H_k(x)=e^{x^2}\Pa{-\frac{\dif}{\dif x}}^ke^{-x^2}=k!\sum^{\Br{\frac
k2}}_{i=0}(-1)^i\frac{(2x)^{k-2i}}{i!(k-2i)!}, \quad k=0,1,2,\ldots
$$
be the $k$-th Hermite polynomial. Denote
$\varphi_k(x)=\frac1{\sqrt{2^kk!\sqrt{\pi}}}e^{-\frac{x^2}2}H_k(x)$.
Then the pdf of a generic eigenvalue of a Hermitian random matrix
from $\rG\rU\rE(n)$ is given by
\begin{eqnarray}\label{eq:density-of-single}
p(x) = \frac1n\sum^{n-1}_{k=0}\varphi^2_k(x) = \varphi^2_n(x) -
\sqrt{1+\frac1n}\varphi_{n-1}(x)\varphi_{n+1}(x),\quad x\in\real.
\end{eqnarray}
\end{prop}

In summary, we get that
\begin{eqnarray}
\mathbb{E}_{\bsH\in \rG\rU\rE(n)}\Br{f(\bsH)} =\Pa{
\int^\infty_{-\infty}f(x)p(x)\dif x}\cdot\I_n.
\end{eqnarray}
Here $p(x)$ is taken from \eqref{eq:density-of-single}.
\begin{prop}[\cite{Haagerup2003}]\label{prop:haagerup}
It holds that
\begin{eqnarray}
\mathbb{E}_{\bsH\in\rG\rU\rE(n)}\Br{e^{t\bsH}} =
\kappa_n(\exp,t)\cdot\I_n,\quad t\in\real,
\end{eqnarray}
where the constant $\kappa_n(\exp,t)$ is given by
\begin{eqnarray*}
\kappa_n(\exp,t)=\int^\infty_{-\infty}e^{tx} p(x)\dif x =
e^{\frac{t^2}4} F\Pa{1-n,2;-\frac{t^2}2},
\end{eqnarray*}
and $F$ is the confluent hypergeometric function defined by
$F(a,c;z)\defeq \sum^\infty_{k=0}\frac{(a)_kz^k}{(c)_kk!}$ for
$a,c,z\in\complex$ such that $c\notin \integer\backslash\natural$.
Here $(a)_k:=\prod^{k-1}_{j=0}(a+j)$ is the Pochhammer notation.
\end{prop}
Furthermore, based on Propositions~\ref{prop:haagerup} and
\ref{th:sum-of-Kgue}, we can easily derive the following identity:
Let $\bsC\defeq\bsA+\bsB$ whose eigenvalues consist of the vector
$\bsc=(c_1,\ldots,c_n)$. Then
\begin{eqnarray}
\mathbb{E}_{\bsA,\bsB\in\rG\rU\rE(n)}\Br{f(\bsA+\bsB)}&=&
\frac1{2^{\frac n2} \pi^{\frac{n^2}2}}\int
f(\bsC)\exp\Pa{-\frac12\Tr{\bsC^2}} [\dif\bsC] \\
&=& \kappa_n\cdot\I_n, \nonumber
\end{eqnarray}
where
\begin{eqnarray*}
\kappa_n \defeq
\frac1n\int_{\real^n}\Tr{f(\widehat\bsc)}p_{\bsC,2}(\bsc)[\dif\bsc].
\end{eqnarray*}
%
%

It is well known that, for any two $n\times n$ Hermitian matrices
$\bsX$ and $\bsY$, the Golden-Thompson inequality holds:
\begin{eqnarray}
\Tr{e^{\bsX+\bsY}}\leqslant \Tr{e^{\bsX}e^{\bsY}}.
\end{eqnarray}
Define the ratio
$$
\alpha_n \defeq
\frac{\mathbb{E}_{\bsX,\bsY\in\rG\rU\rE(n)}\Br{\Tr{e^{\bsX}e^{\bsY}}}}{\mathbb{E}_{\bsX,\bsY\in\rG\rU\rE(n)}\Br{\Tr{e^{\bsX+\bsY}}}}.
$$
Then
\begin{eqnarray} \label{eqn:ratio_alphan}
\alpha_n = \frac{F\Pa{1-n,2;-\frac12}^2}{F\Pa{1-n,2;-1}}.
\end{eqnarray}
For small numbers $n=2,3$ and 4, we can analytically obtain that
$\alpha_2= \frac{25}{24}\approx e^{0.040822}$, $\alpha_3=
\frac{1369}{1248}\approx e^{0.0925383}$ and $\alpha_4=
\frac{130321}{112128}\approx e^{0.15036}$. In fact, we have a more
interesting result, which seems that the leading asymptotic form is
$e^{2(\sqrt{2}-1)\sqrt{n}}$ (see Theorem~\ref{th:asympratio} with
its proof). Thus, we observe that the ratio $\alpha_n$ exponentially
increases with the matrix size $n$. This partly answers the question
proposed by Forrester and Thompson in \cite{Forrester2014}.

In what follows, we will derive the asymptotic form of $\alpha_n$.
In order to do that, we need the explicit relationship between
Laguerre polynomials and Kummer's functions (Kummer's function is
also called confluent hypergeometric function $F(a,c;z)$, mentioned
in Proposition~\ref{prop:haagerup}). Some details about the notions
of Laguerre polynomails and Kummer's function and its relationship
can be found in \cite{Olver2010}. Now we can prove the following
result.
\begin{thrm}\label{th:asympratio}
It holds that
\begin{eqnarray}
\lim_{n\to\infty}\frac{\ln\alpha_n}{\sqrt{n}}= 2(\sqrt{2}-1).
\end{eqnarray}
\end{thrm}

\begin{proof}
Using Eq.(13.6.19) in \cite{Olver2010}, we find the explicit
relationship between Laguerre polynomials and Kummer's functions
given as follows. For an arbitrary natural number $n>0$,
$\alpha\in\real$, and $x\in\real$, we have
\begin{eqnarray}\label{eq:keyasymp}
F(-n,\alpha+1;x) = \frac{n!}{(\alpha+1)_n}L^{(\alpha)}_n(x)=
\binom{n+\alpha}{n}L^{(\alpha)}_n(x),
\end{eqnarray}
where $\binom{n+\alpha}{n}$ is the generalized binomial coefficient.
Next, replacing $n$ by $n-1$ and letting $\alpha=1$ in
Eq.~\eqref{eq:keyasymp}, we get that
\begin{eqnarray*}
F(1-n,2;x) = \frac1n L^{(1)}_{n-1}(x).
\end{eqnarray*}
The Laguerre polynomials' asymptotic behavior for large $n$, but
fixed $\alpha$ and $x>0$, is given by \cite{Borwein2008}:
\begin{eqnarray}\label{eq:asympoflag}
L^{(\alpha)}_n(-x)
&=&\frac{(n+1)^{\frac\alpha2-\frac14}}{2\sqrt{\pi}}\frac{e^{-\frac
x2}}{x^{\frac\alpha2+\frac14}}e^{2\sqrt{(n+1)x}}\Pa{1+O\Pa{\frac1{\sqrt{n+1}}}}.
\end{eqnarray}
From Eq.~\eqref{eq:asympoflag}, we infer that, for large $n$ and
fixed $x>0$,
\begin{eqnarray}\label{eq:key2}
F(1-n,2;-x) = \frac1{2\sqrt{\pi}}\frac{\exp\Pa{2\sqrt{nx}-\frac
x2}}{(nx)^{\frac34}}\Pa{1+O\Pa{\frac1{\sqrt{n}}}}
\end{eqnarray}
Based on the observation, substituting $x=\frac12,1$, respectively,
in Eq.~\eqref{eq:key2} and plugging them into the expression of
$\alpha_n$, we derive that
\begin{eqnarray}
\alpha_n \sim \sqrt{\frac2\pi}n^{-\frac34}e^{2(\sqrt{2}-1)\sqrt{n}}
\end{eqnarray}
for large $n$. Therefore,
\begin{eqnarray*}
\lim_{n\to\infty}\frac{\ln\alpha_n}{\sqrt{n}}= 2(\sqrt{2}-1).
\end{eqnarray*}
This competes the proof.
\end{proof}

\section{Applications in quantum information theory}\label{sect:app}

Recall that a positive semi-definite matrix of unit trace are called
a density matrix which are used to describe a state of quantum
systems. Denote the set of all density matrices of size $n$ by
$\density{\complex^n}$. Let $\bsA=\frac12\rho_1$ and
$\bsB=\frac12\rho_2$, where $\rho_1,\rho_2\in\density{\complex^2}$,
the set of density matrices of all qubit states. Then
$\alpha=\frac12(1-\mu,\mu), \beta=\frac12(1-\nu,\nu)$ and
$\gamma=(1-\lambda,\lambda)$ where $\mu,\nu\in(0,1/2)$ and
$\lambda\in (0,1)$. Thus
\begin{eqnarray}
p(\lambda|\mu,\nu) =\begin{cases}
\frac{\frac12-\lambda}{(\frac12-\mu)(\frac12-\nu)},&\text{if }
\lambda\in[T_0,T_1],\\
\frac{\lambda-\frac12}{(\frac12-\mu)(\frac12-\nu)},&\text{if
}\lambda\in\Br{1-T_1,1-T_0},
\end{cases}
\end{eqnarray}
where $T_0:=T_0(\mu,\nu) = \frac{\mu+\nu}2$ and $T_1:=T_1(\mu,\nu) =
\frac{1-\abs{\mu-\nu}}2$. The normalization of $p(\lambda|\mu,\nu)$
is easily obtained by using mathematical software. We omit it here.

Let $\bsC^\diag=\diag(1-x,x)$ for $x\in[0,1]$. Then
$\bsC^\diag=\bsA^\diag+\bsB^\diag$ since $\bsC=\bsA+\bsB$. Thus
there exist $t,s\in[0,1]$ such that
\begin{eqnarray*}
1-x &=& (1-t)\frac{1-\mu}2+ t\frac \mu2 + (1-s)\frac{1-\nu}2+ s\frac
\nu2,\\
x&=&t\frac{1-\mu}2+ (1-t)\frac \mu2 + s\frac{1-\nu}2+ (1-s)\frac
\nu2.
\end{eqnarray*}
This implies that
$$
1-2x = \Pa{\frac12-\mu}(1-2t) + \Pa{\frac12-\nu}(1-2s).
$$
For fixed $\mu,\nu\in\Br{0,\frac12}$, we see that as a function of
arguments $(t,s)\in[0,1]\times[0,1]$,
$$
\Pa{\frac12-\mu}(1-2t) + \Pa{\frac12-\nu}(1-2s)
$$
has the maximum $1-\mu-\nu$ and the minimum $\mu+\nu-1$. Therefore,
$$
\mu+\nu-1\leqslant 1-2x\leqslant 1-\mu-\nu\Longleftrightarrow
x\in[T_0,1-T_0].
$$
We also have that
\begin{eqnarray*}
q(x|\mu,\nu) =
\frac1{2\Pa{\frac12-\mu}\Pa{\frac12-\nu}}\Pa{\abs{x-T_0}+\abs{x-(1-T_0)}-\abs{x-T_1}-\abs{x-(1-T_1)}}.
\end{eqnarray*}
Furthermore, simplifying it into the following form:
\begin{eqnarray}
q(x|\mu,\nu) =
\frac1{\Pa{\frac12-\mu}\Pa{\frac12-\nu}}\begin{cases}x-T_0,&\text{if
}x\in[T_0,T_1],\\T_1-T_0,&\text{if
}x\in[T_1,1-T_1],\\-x+(1-T_0),&\text{if
}x\in[1-T_1,1-T_0].\end{cases}
\end{eqnarray}
Note that $q(x|\mu,\nu)$ is concentrated on $[T_0,1-T_0]\subset
[0,1]$ and $q(x|\mu,\nu)$ is vanished on
$[0,1]\backslash[T_0,1-T_0]$. The normalization can be checked
directly as follows:
\begin{eqnarray*}
\int^1_0 q(x|\mu,\nu)\dif x &=&
\frac1{\Pa{\frac12-\mu}\Pa{\frac12-\nu}}\\
&&\times\Br{\int^{T_1}_{T_0}\Pa{x-T_0}\dif x + \int^{1-T_1}_{T_1}
(T_1-T_0)\dif x + \int^{1-T_0}_{1-T_1}\Pa{-x+(1-T_0)}\dif x}\\
&=&\frac{\Br{\frac12x^2 - T_0x}^{T_1}_{T_0} + (T_1-T_0)(1-2T_1) +
\Br{-\frac12x^2+(1-T_0)x}^{1-T_0}_{1-T_1}}{\Pa{\frac12-\mu}\Pa{\frac12-\nu}},
\end{eqnarray*}
implying that
\begin{eqnarray}
\int^1_0 q(x|\mu,\nu)\dif x =
\frac{(T_1-T_0)(1-T_1-T_0)}{\Pa{\frac12-\mu}\Pa{\frac12-\nu}} = 1.
\end{eqnarray}

\subsection{Uniform average distance between two orbits}

We consider the following quantum Jensen-Shannon divergence which is
defined by
$$
\rQ\rJ\rS\rD(\rho_1,\rho_2):=
\frac12\Br{\rS(\rho_1||\overline{\rho})+\rS(\rho_2||\overline{\rho})},
$$
where $\overline{\rho}=\frac12\rho_1+\frac12\rho_2$ and
$\rS(\rho_i||\overline{\rho})=\Tr{\rho_i(\ln\rho_i-\ln\overline{\rho})}$
is the so-called relative entropy. Clearly,
$\rQ\rJ\rS\rD(\rho_1,\rho_2)=\rS\Pa{\frac12\rho_1+\frac12\rho_2}-\frac12\rS(\rho_1)-\frac12\rS(\rho_2)$.
We may use our results obtained previously to investigate the
average QJSD between two unitary orbits:
\begin{eqnarray*}
&&\iint\dif\mu_{\mathrm{Haar}}(U)\dif\mu_{\mathrm{Haar}}(V)\rQ\rJ\rS\rD(U\rho_1U^\dagger,V\rho_2V^\dagger)\\
&&=\iint\dif\mu_{\mathrm{Haar}}(U)\dif\mu_{\mathrm{Haar}}(V)\rS\Pa{\frac12U\rho_1U^\dagger+\frac12V\rho_2V^\dagger}
-\frac12\rH_2(\mu)-\frac12\rH_2(\nu)\\
&&=
\int^{T_1}_{T_0}\rH_2(\lambda)\frac{\frac12-\lambda}{(\frac12-\mu)(\frac12-\nu)}\dif\lambda+
\int^{1-T_0}_{1-T_1}\rH_2(\lambda)\frac{\lambda-\frac12}{(\frac12-\mu)(\frac12-\nu)}\dif\lambda-\frac12\rH_2(\mu)-\frac12\rH_2(\nu),
\end{eqnarray*}
where $\rH_2(x):=-x\ln x-(1-x)\ln(1-x)$ is the binary entropy
function for $x\in[0,1]$.

This gives an explicit expression about the uniform average QJSD,
i.e., uniform average quantum Jensen-Shannon divergence, between two
distinctive isospectral quantum states. Define
$$
\overline{\rQ\rJ\rS\rD}(\cO_\mu,\cO_\nu) :=
\iint\dif\mu_{\mathrm{Haar}}(U)\dif\mu_{\mathrm{Haar}}(V)\rQ\rJ\rS\rD(U\rho_1U^\dagger,V\rho_2V^\dagger),
$$
where $\cO_x=\Set{U\diag(1-x,x)U^\dagger: U\in\rU(2)}$ for
$x\in[0,1]$. Thus we see that
\begin{eqnarray}
\overline{\rQ\rJ\rS\rD}(\cO_\mu,\cO_\nu) = \int^1_0
\rH_2(\lambda)p(\lambda|\mu,\nu)\dif \lambda
-\frac12\rH_2(\mu)-\frac12\rH_2(\nu),
\end{eqnarray}
where
\begin{eqnarray}
\int^1_0 \rH_2(\lambda)p(\lambda|\mu,\nu)\dif \lambda &=&
\frac1{(\frac12-\mu)(\frac12-\nu)}\\
&&\times\Br{\int^{T_1}_{T_0}\rH_2(\lambda)\Pa{\frac12-\lambda}\dif\lambda+\int^{1-T_0}_{1-T_1}\rH_2(\lambda)\Pa{\lambda-\frac12}\dif\lambda}.
\end{eqnarray}
Note that by using \textsc{Mathematica}~10.0, we get that
\begin{eqnarray*}
F_0(x)&:=&\int \rH_2(x)\dif x = \frac14(1-x)^2 \Br{2\ln (1-x)-1} -
\frac14x^2(2\ln
x-1)\\
F_1(x)&:=&\int \rH_2(x)x\dif x \\
&=& \frac{1}{4} (1-x)^2\Br{2\ln (1-x)-1} - \frac{1}{9}
(1-x)^3\Br{3\ln (1-x)-1} - \frac19x^3(3\ln x-1).
\end{eqnarray*}
We also find that, for $x\in[0,1]$
\begin{eqnarray}
F_0(1-x)+F_0(x) &=& 0,\label{eq:F-0}\\
F_1(x)-F_1(1-x) &=&F_0(x).\label{eq:F-1}
\end{eqnarray}
Thus,
\begin{eqnarray}
\overline{\rQ\rJ\rS\rD}(\cO_\mu,\cO_\nu) = \frac{\Br{\frac12
F_0(\lambda)-F_1(\lambda)}^{T_1}_{T_0}+\Br{F_1(\lambda)-\frac12
F_0(\lambda)}^{1-T_0}_{1-T_1}}{(\frac12-\mu)(\frac12-\nu)}
-\frac12\rH_2(\mu)-\frac12\rH_2(\nu).
\end{eqnarray}
Denote $\overline{\rQ\rJ\rS\rD}(\cO_\mu,\cO_\nu)$ as
$\varphi(\mu,\nu)$, where $(\mu,\nu)\in\Br{0,\frac12}\times
\Br{0,\frac12}$. That is,
\begin{eqnarray}
\varphi(\mu,\nu) := \frac{\Phi(\mu,\nu)}{(\frac12-\mu)(\frac12-\nu)}
-\frac12\rH_2(\mu)-\frac12\rH_2(\nu),
\end{eqnarray}
where
\begin{eqnarray*}
\Phi(\mu,\nu) &=& \Br{\frac12
F_0(\lambda)-F_1(\lambda)}^{T_1}_{T_0}+\Br{F_1(\lambda)-\frac12
F_0(\lambda)}^{1-T_0}_{1-T_1}\\
&=& \frac12(F_0(T_1) + F_0(1-T_1) - F_0(T_0) - F_0(1-T_0))\\
&&+(F_1(T_0) + F_1(1-T_0) - F_1(T_1) - F_1(1-T_1))
\end{eqnarray*}
and $(\mu,\nu)\in\Br{0,\frac12}\times \Br{0,\frac12}$. Thus using
\eqref{eq:F-0}, we get
\begin{eqnarray}\label{eq:Phi}
\Phi(\mu,\nu) = F_1(T_0) + F_1(1-T_0) - F_1(T_1) - F_1(1-T_1).
\end{eqnarray}
And therefore we have
\begin{eqnarray}
\varphi(\mu,\nu) = \frac{F_1(T_0) + F_1(1-T_0) - F_1(T_1) -
F_1(1-T_1)}{(\frac12-\mu)(\frac12-\nu)}-\frac12\rH_2(\mu)-\frac12\rH_2(\nu).
\end{eqnarray}

The symmetry of $\mu$ and $\nu$ is obvious. The maximum is achieved
at $(0,0)$. That is, $\varphi(0,0)=\frac{1}{3} \ln 2 + \frac{1}{6}$.
And the minimum is $\varphi(\frac{1}{2},\frac{1}{2})=0$.  The case
where $\mu$ or $\nu=0$ corresponds to the pure state, and the case
where $\mu$ or $\nu=\frac{1}{2}$ corresponds to the maximum mixed
state. Hence, the quantum Jensen-Shannon divergence is the largest
in the case of two pure states, and smallest in the case of two
maximally mixed states.
In addition, $\varphi(\frac{1}{2},0)=\varphi(0,\frac{1}{2}) = -\frac{3}{4} \ln \frac{3}{4} = 0.2158 $.

\subsection{Average coherence of uniform mixture of two orbits}

In this subsection, we will see that the derivative principle is
naturally related to quantum coherence of relative entropy, a hot
topic which is intensive studied from various viewpoints
\cite{Baumgratz2014}. Let us recall that the coherence of relative
entropy proposed by Baumgratz \emph{et al} is quantified by the
relative entropy between a given state and the nearest incoherent
state to the original one. Specifically, it is given by
$\sC(\rho)=\rS(\rho_\diag)-\rS(\rho)$. Surprisedly, it is determined
by both the diagonal part and its eigenvalues of this given state.
When we study the quantum coherence of relative entropy of a random
quantum state and its typicality, the derivative principle can be naturally
applied. Recently, indeed,  we considered the average coherence and
its typicality \cite{Zhang2017a,Zhang2017b}. Here we will present
more subtle result about quantum coherence of a random state limited
to the qubit case.

In the following we consider the average coherence of uniform
mixture of two unitary orbits with respective prescribed spectra,
i.e.,
\begin{eqnarray}
&&\overline{\sC}(\cO_\mu,\cO_\nu):=\iint\dif\mu_{\mathrm{Haar}}(U)\dif\mu_{\mathrm{Haar}}(V)
\sC\Pa{\frac12U\rho_1U^\dagger+\frac12V\rho_2V^\dagger} \nonumber\\
&&=\iint\dif\mu_{\mathrm{Haar}}(U)\dif\mu_{\mathrm{Haar}}(V)\Br{\rS\Pa{\Pa{\frac12U\rho_1U^\dagger+
\frac12V\rho_2V^\dagger}_\diag}-\rS\Pa{\frac12U\rho_1U^\dagger+\frac12V\rho_2V^\dagger}} \nonumber\\
&&= \int^1_0\rH_2(x)q(x|\mu,\nu)\dif x -
\int^1_0\rH_2(\lambda)p(\lambda|\mu,\nu)\dif \lambda, \label{eqn:TwoIntegrals}
\end{eqnarray}
where
\begin{eqnarray*}
\int^1_0\rH_2(x)q(x|\mu,\nu)\dif x
=\frac{\Br{F_1(x)-T_0F_0(x)}^{T_1}_{T_0} +
(T_1-T_0)\Br{F_0(x)}^{1-T_1}_{T_1} + \Br{(1-T_0)F_0(x)-F_1(x)
}^{1-T_0}_{1-T_1}}{(\frac12-\mu)(\frac12-\nu)}.
\end{eqnarray*}
Denote $\overline{\sC}(\cO_\mu,\cO_\nu)$ as $\psi(\mu,\nu)$, where
$(\mu,\nu)\in\Br{0,\frac12}\times \Br{0,\frac12}$. That is,
\begin{eqnarray}
\psi(\mu,\nu) := \frac{\Psi(\mu,\nu) -
\Phi(\mu,\nu)}{(\frac12-\mu)(\frac12-\nu)},
\end{eqnarray}
where
\begin{eqnarray*}
\Psi(\mu,\nu) &:=& \Br{F_1(x)-T_0F_0(x)}^{T_1}_{T_0} +
(T_1-T_0)\Br{F_0(x)}^{1-T_1}_{T_1} + \Br{(1-T_0)F_0(x)-F_1(x)
}^{1-T_0}_{1-T_1}\\
&=& T_0 \cdot F_0(T_0) + (1-T_0)\cdot F_0(1-T_0)-T_1 \cdot F_0(T_1)
- (1-T_1)\cdot F_0(1-T_1)   -   \Phi(\mu,\nu)
\end{eqnarray*}
and $\Phi(\mu,\nu)$ is from Eq.~\eqref{eq:Phi} and
$(\mu,\nu)\in\Br{0,\frac12}\times \Br{0,\frac12}$. Thus using
\eqref{eq:F-0} again, we get
\begin{eqnarray}
\Psi(\mu,\nu)     +     \Phi(\mu,\nu) =\Pa{2T_0-1}\cdot F_0(T_0)-
\Pa{2T_1-1}\cdot F_0(T_1).
\end{eqnarray}
This implies that
\begin{eqnarray}
\psi(\mu,\nu) := \frac{(2T_0-1)F_0(T_0) -
(2T_1-1)F_0(T_1)  -  2  \Phi(\mu,\nu) }{(\frac12-\mu)(\frac12-\nu)},
\end{eqnarray}

We can see that $\psi(\mu,\nu)$ is completely determined by two
integrals $\int^1_0\rH_2(x)q(x|\mu,\nu)\dif x$ and
$\int^1_0\rH_2(x)p(x|\mu,\nu)\dif x$ in \eqref{eqn:TwoIntegrals}.
The biggest difference between the two integrals appears at $(0,0)$,
while the difference approaches zero at $(\frac{1}{2},\frac{1}{2})$.
The substraction of the two integrals yields the function
$\psi(\mu,\nu)$. The symmetry of $\psi$ on the parameters $\mu$ and
$\nu$  can be easily checked. The maximum is achieved at $(0,0)$.
Specifically we have $\psi(0,0)= \frac{2}{3} (1-\ln 2) = 0.2046$.
The minimum is the limit at $(\frac{1}{2},\frac{1}{2})$, i.e.,
$\psi(\frac{1}{2},\frac{1}{2}) = 0$. In addition,
$\psi(0,\frac{1}{2}) = \psi(\frac{1}{2},0)= \frac{1}{2} -
\frac{3}{8}\ln 3=0.0880$.

\section{Concluding remarks}\label{sect:end}

In this paper, we investigate a variant of Horn's problem, i.e., we
identify the pdf of the diagonals of the sum of two random Hermitian
matrices with given spectra. We then use it, together with the
\emph{derivative principle}, to re-derive the pdf of the eigenvalues
of the sum of two random Hermitian matrices with given spectra.
Zuber's recent results on the same problem can be also recovered. We then apply these
results further to derive the analytical expressions of eigenvalues
of the sum of two random Hermitian matrices from $\rG\rU\rE(n)$ or
Wishart ensemble by derivative principle, respectively. We also
investigate the statistics of exponential of random matrices and
connect them with Golden-Thompson inequality, and partially answer a
question proposed by Forrester. The results obtained are also
employed to analyze quantum Jensen-Shannon divergence between two
unitary orbits with their prescribed spectra and quantum coherence
of mixture of two unitary orbits. Although these applications in
quantum information theory are just around 2-dimensional space, we
believe our methods used in the present work can be extended to
higher dimensional space.

\subsection*{Acknowledgement}

The authors sincerely thank the editor and the anonymous referees
for their comments, especially the proposal for
Theorem~\ref{th:asympratio}, which led to great improvements of this
paper. The authors also thank Lin Huang for pointing out the
connection with Laguerre polynomial. Lin Zhang is supported by
Zhejiang Provincial Natural Science Foundation of China under grant
no. LY17A010027 and also by National Natural Science Foundation of
China (Nos.11971140, 11701259, 61771174). Hua Xiang is supported by
National Natural Science Foundation of China (No.11571265) and
NSFC-RGC (China-Hong Kong, No.11661161017).

\appendix
\appendixpage
\addappheadtotoc

\section{Zuber's derivation of the pdf}\label{app:A}

\begin{thrm}[J-B. Zuber, \cite{JBZ2017}]\label{th:triple-int}
Assume that two random matrices $A$ and $B$ chosen uniformly on the
unitary orbits $\cU(\bsa)$ and $\cU(\bsb)$, respectively, then the joint
pdf $p(\bsc|\bsa,\bsb)$ of the eigenvalues $\bsc$ of the sum
$\bsC:=\bsA+\bsB$ is given by the following integral
\begin{eqnarray}
p(\bsc|\bsa,\bsb) = \mathrm{const.}\Delta(\bsc)^2\int_{\real^n}[\dif
\bsx]\Delta(\bsx)^2\cI(\bsa,\mathrm{i}\bsx)\cI(\bsb,\mathrm{i}\bsx)\overline{\cI(\bsc,\mathrm{i}\bsx)},
\end{eqnarray}
where
\begin{eqnarray}
\cI(\bsa,\bsb) =
\int_{\rU(n)}\exp\Pa{\Tr{\widehat{\bsa}\bsU{\widehat{\bsb}}\bsU^\dagger}}\dif\mu_{\mathrm{Haar}}(\bsU)
\end{eqnarray}
is the famous Harish-Chandra integral for which the explicit formula
can be written down \cite{Itzykson1980}:
\begin{eqnarray}
\cI(\bsa,\bsb) =
\Pa{\prod^n_{k=1}\Gamma(k)}\frac{\det\Pa{e^{a_ib_j}}}{\Delta(\bsa)\Delta(\bsb)}.
\end{eqnarray}
\end{thrm}

\begin{prop}
The pdf of eigenvalues $\bsc$, given $\bsa$ and $\bsb$, is given by
\begin{eqnarray}
p(\bsc|\bsa,\bsb) &=&
\frac{\prod^n_{k=1}\Gamma(k)}{(2\pi)^n(n!)^2\mathrm{i}^{\binom{n}{2}}}
\frac{\Delta(\bsc)}{\Delta(\bsa)\Delta(\bsb)}\int\frac{[\dif
\bsx]}{\Delta(\bsx)}
\det\Pa{e^{\mathrm{i}x_ia_j}}\det\Pa{e^{\mathrm{i}x_ib_j}}\det\Pa{e^{-\mathrm{i}x_ic_j}} \\
&=&
\frac{\prod^n_{k=1}\Gamma(k)}{(2\pi)^nn!\mathrm{i}^{\binom{n}{2}}}
\frac{\Delta(\bsc)}{\Delta(\mathbf{a})\Delta(\mathbf{b})}\int\frac{[\dif
\bsx]}{\Delta(\bsx)}
\det\Pa{e^{\mathrm{i}x_ia_j}}\det\Pa{e^{\mathrm{i}x_ib_j}}\prod^n_{k=1}e^{-\mathrm{i}x_kc_k}.
\label{eqn:p(c;a,b)withDet}
\end{eqnarray}
\end{prop}

\begin{proof}
Note that we can use permutational symmetry of the integrand to
replace $\det\Pa{e^{-\mathrm{i}x_ic_j}}$ by
$n!e^{-\mathrm{i}\Inner{\bsx}{\bsc}}=n!\prod^n_{k=1}e^{-\mathrm{i}x_kc_k}$.
\end{proof}

Note that \eqref{eqn:p(c;a,b)withDet} is exactly the same as
\eqref{eqn:q(c|ab)inThm}, which is derived from $q(\bsc|\bsa,\bsb)$
and the derivative principle in Section \ref{sect:der-pin}.

\begin{cor}
The pdf of eigenvalues $\bsc$, given $\bsa$ and $\bsb$, is given by
\begin{eqnarray}
p(\bsc|\bsa,\bsb) =
\frac{\prod^n_{k=1}\Gamma(k)}{(2\pi)^{n-1}n!\mathrm{i}^{\binom{n}{2}}}
\delta\Pa{\sum^n_{k=1}(a_k+b_k-c_k)}\frac{\Delta(\bsc)}{\Delta(\bsa)\Delta(\bsb)}J(\bsa,\bsb:\bsc),
\end{eqnarray}
where $J(\bsa,\bsb:\bsc)$ is given by
\begin{eqnarray}
J(\bsa,\bsb:\bsc) = \sum_{\sigma,\tau\in S_n}
\sign(\sigma\tau)\int_{\real^{n-1}} \frac{[\dif
\bsu]}{\widetilde\Delta(\bsu)}\prod^{n-1}_{k=1}\exp\Br{\mathrm{i}u_kA_k(\sigma,\tau)},
\end{eqnarray}
and where
$$
A_k(\sigma,\tau)= \sum^k_{j=1}\Pa{a_{\sigma(j)}+b_{\tau(j)}-c_j}
-\frac kn\sum^n_{j=1}(a_j+b_j-c_j).
$$
\end{cor}

Although an approach toward the integral in
Theorem~\ref{th:triple-int} is sketched by Zuber in \cite{JBZ2017},
we  choose to reconstruct the details for reader's convenience. Note
that there are a few different tricks from the way sketched by
Zuber.

\begin{proof}
We simplify the two determinants in the integral of
\eqref{eqn:p(c;a,b)withDet}. Then applying
\eqref{eqn:CorProof_Det}, we get
\begin{eqnarray*}
\det\Pa{e^{\mathrm{i}x_ia_j}}\det\Pa{e^{\mathrm{i}x_ib_j}}&=&e^{\mathrm{i}\bar{x}\sum^n_{k=1}(a_k+b_k)}\sum_{\sigma,\tau\in S_n} \sign(\sigma\tau)\\
&&\times
\prod^{n-1}_{k=1}\exp\Br{\mathrm{i}(x_k-x_{k+1})\Pa{\sum^k_{j=1}(a_{\sigma(j)}+b_{\tau(j)})}
-\frac kn\sum^n_{j=1}(a_j+b_j)}.
\end{eqnarray*}
We switch to the last term in the integral of
\eqref{eqn:p(c;a,b)withDet}. Thus, by using \eqref{eq:ip-formula}
\begin{eqnarray*}
\prod^n_{k=1}e^{-\mathrm{i}x_kc_k}=
\exp\Br{-\mathrm{i}\bar{x}\sum^n_{k=1}c_k}
\exp\Br{-\mathrm{i}\sum^{n-1}_{k=1}(x_k-x_{k+1})\Pa{\sum^k_{j=1}c_j
- \frac kn\sum^k_{j=1}c_j}}.
\end{eqnarray*}
Therefore, we have
\begin{eqnarray}
&&\det\Pa{e^{\mathrm{i}x_ia_j}}\det\Pa{e^{\mathrm{i}x_ib_j}}\prod^n_{k=1}e^{-\mathrm{i}x_kc_k}\\
&&=e^{\mathrm{i}\bar{x}\sum^n_{k=1}(a_k+b_k-c_k)}\sum_{\sigma,\tau\in
S_n}
\sign(\sigma\tau)\prod^{n-1}_{k=1}\exp\Br{\mathrm{i}(x_k-x_{k+1})A_k(\sigma,\tau)}.
\end{eqnarray}
where
$$
A_k(\sigma,\tau)= \sum^k_{j=1}\Pa{a_{\sigma(j)}+b_{\tau(j)}-c_j}
-\frac kn\sum^n_{j=1}(a_j+b_j-c_j).
$$
It follows that
\begin{eqnarray*}
&&\int_{\real^n}\frac{[\dif \bsx]}{\Delta(\bsx)}
\det\Pa{e^{\mathrm{i}x_ia_j}}\det\Pa{e^{\mathrm{i}x_ib_j}}\prod^n_{k=1}e^{-\mathrm{i}x_kc_k}\\
&&=\int_\real e^{\mathrm{i}\bar{x}\sum^n_{k=1}(a_k+b_k-c_k)}\dif\bar
x\sum_{\sigma,\tau\in S_n} \sign(\sigma\tau)\int_{\real^{n-1}}
\frac{[\dif
\bsu]}{\widetilde\Delta(\bsu)}\prod^{n-1}_{k=1}\exp\Br{\mathrm{i}u_kA_k(\sigma,\tau)},
\end{eqnarray*}
where
$$
\widetilde \Delta(\bsu) := \prod_{1\leqslant i\leqslant j-1\leqslant
n-1}(u_i+u_{i+1}+\cdots +u_{j-1}).
$$
Therefore, we see that
\begin{eqnarray*}
&&\int_{\real^n}\frac{[\dif \bsx]}{\Delta(\bsx)}
\det\Pa{e^{\mathrm{i}x_ia_j}}\det\Pa{e^{\mathrm{i}x_ib_j}}\prod^n_{k=1}e^{-\mathrm{i}x_kc_k}\\
&&=2\pi\delta\Pa{\sum^n_{k=1}(a_k+b_k-c_k)}\sum_{\sigma,\tau\in S_n}
\sign(\sigma\tau)\int_{\real^{n-1}} \frac{[\dif
\bsu]}{\widetilde\Delta(\bsu)}\prod^{n-1}_{k=1}\exp\Br{\mathrm{i}u_kA_k(\sigma,\tau)}.
\end{eqnarray*}
We are done.
\end{proof}



\end{document}